\documentclass[letterpaper, 10 pt, conference]{ieeeconf}  

\usepackage{amsmath}
\usepackage{graphicx}
\usepackage{cite}
\usepackage{balance}
\usepackage{amssymb,amsfonts}

\usepackage{lipsum}
\usepackage{mathtools}
\usepackage{cuted}

\newtheorem{theorem}{\it Theorem}

\newtheorem{definition}{\it Definition}

\newtheorem{corollary}{\it Corollary}

\IEEEoverridecommandlockouts                              

\title{\LARGE \bf
	Fundamental Limits of Obfuscation for Linear Gaussian Dynamical Systems: An Information-Theoretic Approach
}

\author{Song Fang$^{1}$ and Quanyan Zhu$^{1}$
\thanks{$^{1}$ Song Fang and Quanyan Zhu are with the Department of Electrical and Computer Engineering, New York University, New York, USA
        {\tt\small song.fang@nyu.edu; quanyan.zhu@nyu.edu}}%
}

\begin{document}

\maketitle
\thispagestyle{empty}
\pagestyle{empty}

\begin{abstract}

In this paper, we study the fundamental limits of obfuscation in terms of privacy-distortion tradeoffs for linear Gaussian dynamical systems via an information-theoretic approach. Particularly, we obtain analytical formulas that capture the fundamental privacy-distortion tradeoffs when privacy masks are to be added to the outputs of the dynamical systems, while indicating explicitly how to design the privacy masks in an optimal way: The privacy masks should be colored Gaussian with power spectra shaped specifically based upon the system and noise properties.

\end{abstract}

\section{INTRODUCTION}

Privacy in dynamical systems (see, e.g.,  \cite{venkitasubramaniam2015information, tanaka2017directed, han2018privacy, nekouei2019information, lu2019control, le2020differential, farokhi2020privacy} and the references therein) is a critical issue that is becoming more and more important nowadays, due to the ever-increasing amount of applications of cyber-physical systems. On the other hand, information-theoretic privacy (see, e.g., \cite{wyner1975wire, bloch2008wireless, liang2009information, liu2010securing, el2011network, bloch2011physical, sankar2011competitive, venkitasubramaniam2015information, han2016event, schaefer2017information, tanaka2017directed, nekouei2019information, farokhi2020privacy} and the references therein) is a fundamental privacy concept, whereas arguably the most commonly used information-theoretic measure of privacy leakage is mutual information (see, e.g., \cite{wyner1975wire, bloch2008wireless, liang2009information, liu2010securing, el2011network, bloch2011physical, sankar2011competitive, venkitasubramaniam2015information, han2016event, schaefer2017information, tanaka2017directed, nekouei2019information, farokhi2020privacy} and the references therein). Recent progress on information-theoretic privacy of dynamical systems includes, e.g., \cite{venkitasubramaniam2015information, tanaka2017directed, farokhi2020privacy} (see also the references therein). 
On the other hand, information-theoretic formulations of the privacy-distortion tradeoff (or, privacy-utility tradeoff) problems have been considered in, e.g., \cite{rebollo2009t, du2012privacy, sankar2012smart, sankar2013utility, makhdoumi2013privacy, CLarxiv} (see also the references therein) for static or time-series data; in this paper, we generalize the formulation to dynamical systems.




Particularly, we focus on analyzing the fundamental information-theoretic privacy-distortion tradeoffs for linear Gaussian dynamical systems. Consider the scenario in which a privacy mask is to be added to the output of a dynamical system, leading to a masked version of the output that is to be revealed to the public. Accordingly, we may view the state of the system as the private information, the original output of the system as the useful information, and the masked output as the disclosed information.
The information-theoretic privacy leakage is then  defined as the mutual information between the state of the system and the masked output, while the distortion is defined between the original output of the system and the masked output. As such, the following questions  naturally arise: What is the fundamental tradeoff between the state privacy leakage and the output distortion led to by the privacy mask? (Given a certain distortion constraint, what is the minimum privacy leakage? Or equivalently, given a certain privacy level, what is the minimum degree of distortion?) 
How to design the privacy mask in an optimal way?

The main contribution of this paper is to provide analytical solutions to the aforementioned questions via an information-theoretic approach. More specifically, by viewing the dynamical system with privacy masks as a ``virtual channel", we derive analytical formulas that capture the fundamental privacy-distortion tradeoffs, while indicating explicitly how to design the privacy masks in an optimal way: The privacy masks should be colored Gaussian with power spectra shaped specifically based upon the system and noise properties. In addition, the optimal solution mandates that more power should be delivered to frequencies at which the ``channel input" power spectra are larger, when above a threshold, whereas below that threshold, no power shall be allocated. In this sense, this solution may be viewed as a ``thresholded obfuscating" power allocation policy. 
We also present further discussions on the implications of the obtained results, including the connection with conditional entropy, the comparison with i.i.d. Gaussian privacy masks, and the investigation of some related problems.



The remainder of the paper is organized as follows. Section~II
introduces the technical preliminaries. Section~III presents the fundamental privacy-distortion tradeoffs for linear Gaussian dynamical systems, as well as solutions for the optimal privacy mask design. Conclusions are given in Section~IV.

%

\section{PRELIMINARIES}

Throughout the paper, we consider real-valued continuous random variables and random vectors, as well as discrete-time stochastic processes. All random variables, random vectors, and stochastic processes are assumed to be zero-mean. We represent random variables and random vectors using boldface letters. Given a stochastic process $\left\{ \mathbf{x}_{k}\right\}$, we denote the sequence $\mathbf{x}_0,\ldots,\mathbf{x}_{k}$ by $\mathbf{x}_{0,\ldots,k}$ for simplicity. The logarithm is with base $2$. A stochastic process $\left\{ \mathbf{x}_{k}\right\}, \mathbf{x}_k \in  \mathbb{R}^m$, is said to be  stationary if $ R_{\mathbf{x}}\left( i,k\right)=\mathbb{E}\left[  \mathbf{x}_i \mathbf{x}_{i+k}^{\mathrm{T}} \right]$ depends only on $k$, and can thus be denoted as  $R_{\mathbf{x}}\left( k\right)$ for simplicity. The power spectrum of a stationary  process $\left\{ \mathbf{x}_{k} \right\}, \mathbf{x}_{k} \in \mathbb{R}^m$, is defined as
\begin{flalign}
\Phi_{\mathbf{x}}\left( \omega\right)
=\sum_{k=-\infty}^{\infty} R_{\mathbf{x}}\left( k\right) \mathrm{e}^{-\mathrm{j}\omega k}. \nonumber
\end{flalign}
Particularly when $m=1$, $\Phi_{\mathbf{x}}\left( \omega\right)$ is denoted as $S_{\mathbf{x}}\left( \omega\right)$, and the variance of $\left\{ \mathbf{x}_{k}\right\}, \mathbf{x}_{k} \in \mathbb{R}$, is given by
\begin{flalign}
\sigma_{\mathbf{x}}^2
= \mathbb{E}\left[ \mathbf{x}_k^2 \right]
= \frac{1}{2\pi}\int_{-\pi}^{\pi} S_{\mathbf{x}}\left(\omega \right) \mathrm{d}  \omega. \nonumber
\end{flalign}


Entropy and mutual information are the most basic notions in information theory \cite{Cov:06}, which we introduce below.

\begin{definition} The differential entropy of a random vector $\mathbf{x} \in \mathbb{R}^m$ with density $p_{\mathbf{x}} \left(x\right)$ is defined as
	\begin{flalign}
	h\left( \mathbf{x} \right)
	=-\int p_{\mathbf{x}} \left(x\right) \log p_{\mathbf{x}} \left(x\right) \mathrm{d} x. \nonumber
	\end{flalign}
	The conditional differential entropy of random vector $\mathbf{x} \in \mathbb{R}^{m}$ given random vector $\mathbf{y} \in \mathbb{R}^{n}$ with joint density $p_{\mathbf{x}, \mathbf{y}} \left(x,y\right)$ and conditional density $p_{\mathbf{x} | \mathbf{y}} \left(x,y\right)$ is defined as
	\begin{flalign}
	h\left(\mathbf{x}\middle|\mathbf{y}\right)
	=-\int p_{\mathbf{x}, \mathbf{y}} \left(x,y\right)\log p_{\mathbf{x} | \mathbf{y}} \left(x,y\right) \mathrm{d}x\mathrm{d}y. \nonumber
	\end{flalign}
	The mutual information between random vectors $\mathbf{x} \in \mathbb{R}^{m}, \mathbf{y} \in \mathbb{R}^{n}$ with densities $p_{\mathbf{x}} \left(x\right)$, $p_{\mathbf{y}} \left( y \right) $ and joint density $p_{\mathbf{x}, \mathbf{y}} \left(x,y\right)$ is defined as
	\begin{flalign}
	I\left(\mathbf{x};\mathbf{y}\right)
	=\int p_{\mathbf{x}, \mathbf{y}} \left(x,y\right) \log \frac{p_{\mathbf{x}, \mathbf{y}} \left(x,y\right)}{p_{\mathbf{x}} \left(x\right) p_{\mathbf{y}} \left( y \right) }\mathrm{d}x\mathrm{d}y. \nonumber
	\end{flalign}
	The entropy rate of a stochastic process $\left\{ \mathbf{x}_{k}\right\},\mathbf{x}_{k}  \in \mathbb{R}^m$, is defined as
	\begin{flalign}
	h_\infty \left(\mathbf{x}\right)=\limsup_{k\to \infty} \frac{h\left(\mathbf{x}_{0,\ldots,k}\right)}{k+1}. \nonumber
	\end{flalign}
	The mutual information rate between two stochastic processes $\left\{ \mathbf{x}_{k}\right\},\mathbf{x}_{k}  \in \mathbb{R}^m$, and $\left\{ \mathbf{y}_{k}\right\},\mathbf{y}_{k}  \in \mathbb{R}^n$, is defined as
	\begin{flalign}
	I_{\infty}\left(\mathbf{x};\mathbf{y}\right)
	=\limsup_{k\to \infty} \frac{I \left(\mathbf{x}_{0,\ldots,k}; \mathbf{y}_{0,\ldots,k}\right)}{k+1}. \nonumber
	\end{flalign}
\end{definition}

\vspace{3mm}

Properties of these notions can be found in, e.g., \cite{Cov:06, yeung2008information, el2011network, fang2017towards}.

\section{FUNDAMENTAL LIMITS OF OBFUSCATION AND OPTIMAL PRIVACY MASK DESIGN}

\begin{figure}
	\vspace{-3mm}
	\begin{center}
		\includegraphics [width=0.5\textwidth]{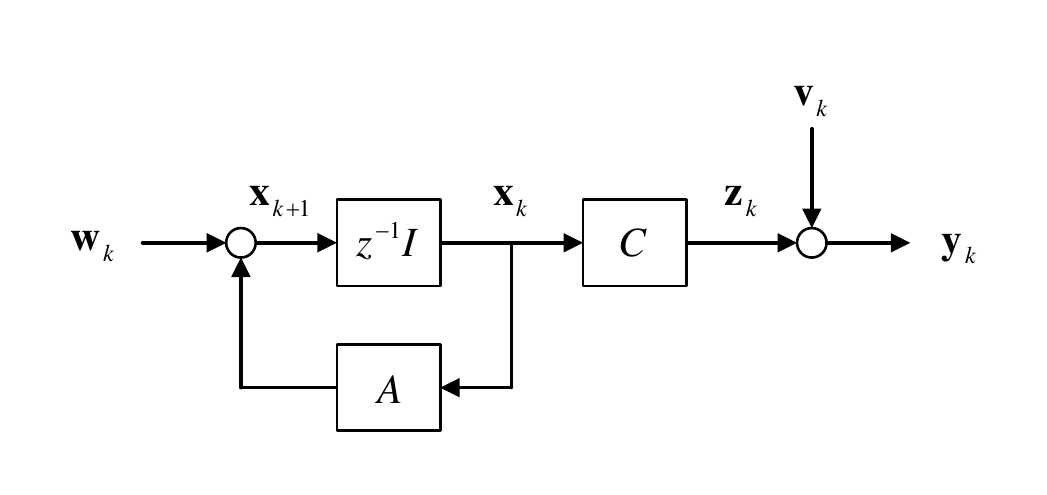}
		\vspace{-9mm}
		\caption{A Dynamical System.}
		\label{system1}
	\end{center}
	\vspace{-6mm}
\end{figure}

In this section, we examine the fundamental limits of obfuscation as well as the optimal privacy mask design for linear Gaussian dynamical systems.
Specifically, consider the dynamical system depicted in Fig.~\ref{system1}  with state-space model given by
\begin{flalign}
\left\{ \begin{array}{rcl}
\mathbf{x}_{k+1} & = & A \mathbf{x}_{k} +\mathbf{w}_k,\\
\mathbf{y}_{k} & = & C \mathbf{x}_{k} +\mathbf{v}_k,
\end{array} \right. \nonumber
\end{flalign}
where $\mathbf{x}_{k} \in \mathbb{R}^{m}$ is the system state, $\mathbf{y}_{k} \in \mathbb{R}$ is the system output, $\mathbf{w}_{k} \in \mathbb{R}^{m}$ is the process noise, and $\mathbf{v}_{k} \in \mathbb{R}$ is the measurement noise.
The system matrices are $ {A} \in \mathbb{R}^{m \times m}$ and $ {C} \in \mathbb{R}^{1 \times m}$; in this paper, we assume that $A$ is stable.
Suppose that $\left\{ \mathbf{w}_{k} \right\}$ and $\left\{ \mathbf{v}_{k} \right\}$ are stationary white Gaussian with covariance matrix $W$ and variance $\sigma_{\mathbf{v}}^2$, respectively. Furthermore,  $\left\{ \mathbf{w}_{k} \right\}$, $\left\{ \mathbf{v}_{k} \right\}$, and $\mathbf{x}_{0}$ are assumed to be mutually independent. It can be verified that the power spectrum of 
$\left\{ \mathbf{z}_{k} \right\}$  
is given by \cite{Pap:02}
\begin{flalign} 
S_{\mathbf{z}} \left( \omega \right)
=	C \left( \mathrm{e}^{\mathrm{j} \omega} I - A \right)^{-1} W \left( \mathrm{e}^{- \mathrm{j} \omega} I - A \right)^{-\mathrm{T}} C^{\mathrm{T}}.
\end{flalign} 

Consider then the scenario that a privacy mask $\left\{ \mathbf{n}_{k} \right\}, \mathbf{n}_{k} \in \mathbb{R}$, is to be added to the output of the system $\left\{ \mathbf{y}_{k} \right\}$ to protect the privacy of the system state $\left\{ \mathbf{x}_{k} \right\}$, resulting in a masked output $\left\{ \widehat{\mathbf{y}}_{k} \right\}$; see the depiction in Fig.~\ref{system2}. Suppose that the privacy mask $\left\{ \mathbf{n}_{k} \right\}$ is independent of $\left\{ \mathbf{w}_{k} \right\}$, $\left\{ \mathbf{v}_{k} \right\}$, and $\mathbf{x}_{0}$; consequently, $\left\{ \mathbf{n}_{k} \right\}$ is independent of $\left\{ \mathbf{x}_{k} \right\}$ and $\left\{ \mathbf{z}_{k} \right\}$ as well. State alternatively, $\left\{ \mathbf{x}_{k} \right\}$ may be viewed as the private information, $\left\{ \mathbf{y}_{k} \right\}$ may be viewed as the useful information, and $\left\{ \widehat{\mathbf{y}}_{k} \right\}$ may be viewed as the information to be disclosed to the public. The following questions then naturally arise: What is the fundamental tradeoff between the state privacy leakage and the output distortion led to by the privacy mask? How to design the privacy mask in an optimal way?

The following theorem, as the main result of this paper, answers the questions raised above.

\begin{theorem} \label{privacy1}
	Consider the dynamical system with privacy masks depicted in Fig.~\ref{system2}. Suppose that the properties of $\left\{ \mathbf{n}_{k} \right\}$ can be designed subject to an output distortion constraint 
	\begin{flalign}
	\mathbb{E} \left[ \left( \mathbf{y}_k - \widehat{\mathbf{y}}_{k} \right)^2 \right] \leq D.
	\end{flalign} 
	Then, in order to minimize the information leakage rate (from the state to the masked output)
	\begin{flalign}
	I_{\infty} \left( \mathbf{x} ; \widehat{\mathbf{y}} \right),
	\end{flalign}
	the noise
	$\left\{ \mathbf{n}_{k} \right\}$ should be chosen as a stationary colored Gaussian process. In addition, the power spectrum of $\left\{ \mathbf{n}_{k} \right\}$ should be chosen as
	\begin{flalign} \label{spectrum1}
	N \left( \omega \right) = \left\{ \frac{\eta}{2 \left[ 1 + \sqrt{1 + \frac{\eta}{S_{\mathbf{z}} \left( \omega \right)} } \right]} - \sigma_{\mathbf{v}}^2 \right\}^{+},
	\end{flalign} 
	where $\eta \geq 0$ satisfies
	\begin{flalign}
	&\frac{1}{2 \pi} \int_{0}^{2 \pi} N \left( \omega \right) \mathrm{d} \omega \nonumber \\
	&~~~~ = \frac{1}{2 \pi} \int_{0}^{2 \pi} \left\{ \frac{\eta}{2 \left[ 1 + \sqrt{1 + \frac{\eta}{S_{\mathbf{z}} \left( \omega \right)} } \right]} - \sigma_{\mathbf{v}}^2 \right\}^{+} \mathrm{d} \omega
	= D.
	\end{flalign} 
	Herein, 
	\begin{flalign}
	\left\{ x \right\}^{+}
	= \begin{cases}
	x, & \text{if}~x > 0;\\
	0, & \text{if}~x\leq 0.
	\end{cases}  \nonumber
	\end{flalign}
	Correspondingly, the minimum information leakage rate is given by
	\begin{flalign} \label{leakage1}
	&\inf_{\mathbb{E} \left[ \left( \mathbf{y}_k - \widehat{\mathbf{y}}_{k} \right)^2 \right] \leq D} I_{\infty} \left( \mathbf{x} ; \widehat{\mathbf{y}} \right) \nonumber \\
	&~~~~ = \frac{1}{2 \pi} \int_{0}^{2 \pi} \log \sqrt{ 1 + \frac{S_{\mathbf{z}} \left( \omega \right)}{N \left( \omega \right) + \sigma_{\mathbf{v}}^2} } \mathrm{d} \omega.
	\end{flalign} 
\end{theorem}

\begin{figure}
	\vspace{-3mm}
	\begin{center}
		\includegraphics [width=0.5\textwidth]{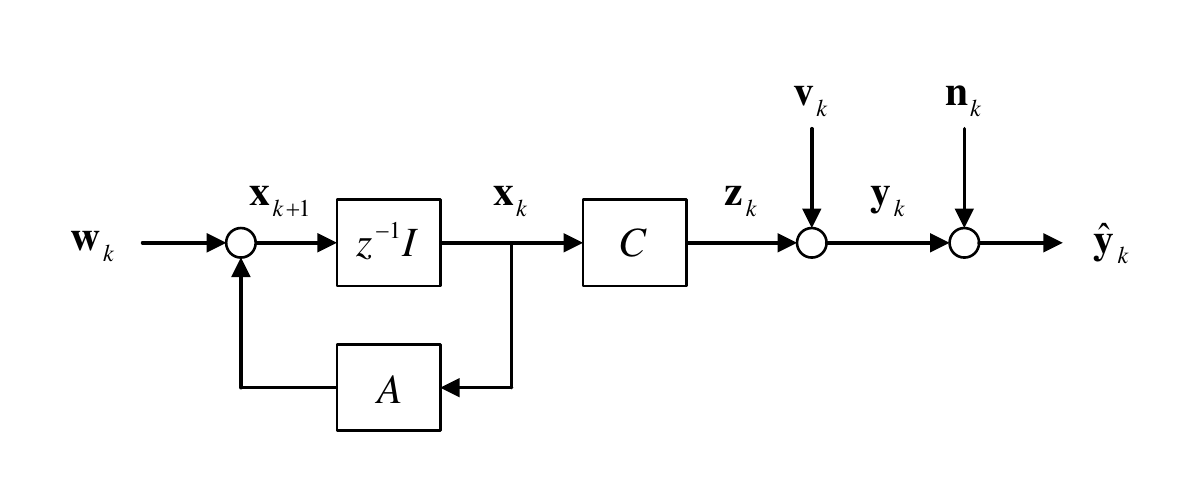}
		\vspace{-9mm}
		\caption{A Dynamical System with Privacy Mask.}
		\label{system2}
	\end{center}
	\vspace{-6mm}
\end{figure}

\vspace{3mm}

\begin{proof}
	Note first that the system in Fig.~\ref{system2} may be viewed as a ``virtual channel" (see also Section~\ref{virtual}) modeled as
	\begin{flalign}
	\widehat{\mathbf{y}}_{k} 
	= \mathbf{y}_{k} + \mathbf{n}_{k}
	= \mathbf{z}_{k} +
	\mathbf{v}_{k} + \mathbf{n}_{k}. \nonumber
	\end{flalign}
	Note then that the distortion constraint $
	\mathbb{E} \left[ \left( \mathbf{y}_k - \widehat{\mathbf{y}}_{k} \right)^2 \right] \leq D $
	is equivalent to being with a power constraint $\mathbb{E} \left[ \mathbf{n}_k^2 \right] \leq D$, since $\widehat{\mathbf{y}}_{k} = \mathbf{y}_{k} + \mathbf{n}_k $ and thus $\left( \mathbf{y}_k - \widehat{\mathbf{y}}_{k} \right)^2 = \mathbf{n}_k^2$.
	
	We start by considering the case of a finite number of parallel (dependent) channels with
	\begin{flalign}
	\widehat{\mathbf{y}} = \mathbf{z} + \mathbf{v} + \mathbf{n}, \nonumber
	\end{flalign}
	where $ \mathbf{z},\mathbf{v},\mathbf{n}, \widehat{\mathbf{y}} \in  \mathbb{R}^m$, while $ \mathbf{z} $, $ \mathbf{v} $, and $ \mathbf{n} $ are mutually independent. In addition, suppose that $\mathbf{z}$ and $\mathbf{v}$ are Gaussian with covariance matrices $\Sigma_{\mathbf{z}}$ and $\Sigma_{\mathbf{n}}$, respectively, and the noise power constraint is given by 
	\begin{flalign} 
	\mathrm{tr} \left(\Sigma_{\mathbf{n} }  \right)
	= \mathbb{E} \left[ \sum_{i=1}^{m}
	\mathbf{n}^{2}\left( i \right) \right] \leq D. \nonumber
	\end{flalign}
	where $\mathbf{n} \left( i \right)$ denotes the $i$-th element of $\mathbf{n}$. Note in particular that herein the elements $\mathbf{v}$, denoted as $\mathbf{v} \left( i \right), i = 1, \ldots, m$, are assumed to be i.i.d. with variance $\sigma_{\mathbf{v}}^2$ and thus $\Sigma_{\mathbf{v}} = \sigma_{\mathbf{v}}^2 I$.
	(Note also that what was described above does not reduce to the case of the channel model considered in \cite{CLarxiv}, due to the presence of $\mathbf{v}$; particularly, $\mathbf{v}$ cannot be merged into $\mathbf{z}$ nor $\mathbf{n}$. Instead, what is considered herein can be viewed as a generalized channel model of that in \cite{CLarxiv}; see Section~\ref{virtual} for further discussions on this.)
	In addition, since $\mathbf{z}$, $ \mathbf{v}$, and $ \mathbf{n}$ are mutually independent, while noting that $\mathbf{z}$ is a function of $\mathbf{x}$, we have
	\begin{flalign}
	I\left(\mathbf{x} ;\widehat{\mathbf{y}} \right)
	&=h\left(\widehat{\mathbf{y}} \right)-h\left(\widehat{\mathbf{y}} \middle| \mathbf{x} \right)
	=h\left(\widehat{\mathbf{y}} \right)-h\left(\mathbf{v} + \mathbf{n}  \middle| \mathbf{x} \right) \nonumber \\ & =h\left(\widehat{\mathbf{y}} \right)-h\left(\mathbf{z} + \mathbf{v} + \mathbf{n}  \middle| \mathbf{x} \right) =h\left(\widehat{\mathbf{y}} \right)-h\left(\mathbf{v} + \mathbf{n} \right), \nonumber
	\end{flalign}
	while 
	\begin{flalign}
	\Sigma_{\mathbf{v} +\mathbf{n} }
	=\Sigma_{\mathbf{v}}
	+\Sigma_{\mathbf{n}}, \nonumber
	\end{flalign}
	and 
	\begin{flalign}
	\Sigma_{\widehat{\mathbf{y}}}
	=\Sigma_{\mathbf{z} +\mathbf{v} +\mathbf{n} }
	=\Sigma_{\mathbf{z}}
	+\Sigma_{\mathbf{v}}
	+\Sigma_{\mathbf{n}}. \nonumber
	\end{flalign}
	Meanwhile, it may be verified that the minimum of $I \left(\mathbf{x}; \widehat{\mathbf{y}} \right)$ is achieved if $\mathbf{v} + \mathbf{n}$ is Gaussian (see, e.g., Section~11.9 of \cite{yeung2008information}), that is, if $\mathbf{n}$ is Gaussian (since $\mathbf{v}$ is assumed to be Gaussian). Particularly when $\mathbf{n}$ is Gaussian, $\widehat{\mathbf{y}}$ is also Gaussian, and it holds that
	\begin{flalign}
	I \left(\mathbf{x}; \widehat{\mathbf{y}} \right)
	&=h \left(\widehat{\mathbf{y}} \right)-h\left(\mathbf{v} + \mathbf{n} \right) \nonumber \\
	&=\frac{1}{2}\log \left[ \left(2\pi \mathrm{e} \right)^m  \det  \Sigma_{\widehat{\mathbf{y}}} \right]-\frac{1}{2}\log \left[ \left(2\pi \mathrm{e} \right)^m  \det  \Sigma_{\mathbf{v} + \mathbf{n}}\right] \nonumber \\
	&=\frac{1}{2}\log \frac{ \det  \Sigma_{\widehat{\mathbf{y}}}}{ \det  \Sigma_{\mathbf{v} + \mathbf{n}}}
	=\frac{1}{2}\log \frac{ \det \left( \Sigma_{\mathbf{z}}+\Sigma_{\mathbf{v}} + \Sigma_{\mathbf{n}}\right)}{ \det  \left( \Sigma_{\mathbf{v}} + \Sigma_{\mathbf{n}}\right) } \nonumber \\
	&=\frac{1}{2}\log \frac{ \det\left( \Sigma_{\mathbf{v}} + \Sigma_{\mathbf{n}} +U_{\mathbf{z}}\Lambda_{\mathbf{z}}U^{\mathrm{T}}_{\mathbf{z}}
		\right)}{ \det  \left( \Sigma_{\mathbf{v}} + \Sigma_{\mathbf{n}} 
		\right)} \nonumber \\ &=\frac{1}{2}\log \frac{ \det\left( \overline{\Sigma}_{\mathbf{v}} + \overline{\Sigma}_{\mathbf{n}} +\Lambda_{\mathbf{z}}
		\right)}{\det \left( \overline{\Sigma}_{\mathbf{v}} + \overline{\Sigma}_{\mathbf{n}} 
		\right)}, \nonumber
	\end{flalign}
	where $U_{\mathbf{z}} \Lambda_{\mathbf{z}} U^{\mathrm{T}}_{\mathbf{z}} $ denotes the eigen-decomposition of $\Sigma_{\mathbf{z}}$ with \begin{flalign}\Lambda_{\mathbf{z}} = \mathrm{diag} \left( \lambda_{1}, \ldots, \lambda_{m} \right), \nonumber
	\end{flalign}
	while  
	\begin{flalign}
	\overline{\Sigma}_{\mathbf{v} }=U^{\mathrm{T}}_{\mathbf{z}}\Sigma_{\mathbf{v}} U_{\mathbf{z}} = U^{\mathrm{T}}_{\mathbf{z}} \sigma_{\mathbf{v}}^2  U_{\mathbf{z}} = \sigma_{\mathbf{v}}^2 U^{\mathrm{T}}_{\mathbf{z}}  U_{\mathbf{z}} = \sigma_{\mathbf{v}}^2 I, \nonumber
	\end{flalign} and $\overline{\Sigma}_{\mathbf{n} }=U^{\mathrm{T}}_{\mathbf{z}}\Sigma_{\mathbf{n}} U_{\mathbf{z}} $. Note that
	\begin{flalign}
	\mathrm{tr} \left( \overline{\Sigma}_{\mathbf{n}} \right)
	&=\mathrm{tr} \left(U^{\mathrm{T}}_{\mathbf{z}}\Sigma_{\mathbf{n} } U_{\mathbf{z}} \right)
	= \mathrm{tr} \left(U_{\mathbf{z}} U^{\mathrm{T}}_{\mathbf{z}}\Sigma_{\mathbf{n} }  \right) =\mathrm{tr} \left(\Sigma_{\mathbf{n} }  \right) \nonumber \\
	&= \mathbb{E} \left[ \sum_{i=1}^{m}\mathbf{n}^{2} \left( i \right) \right] \leq D.
	\nonumber
	\end{flalign}
	As such,
	\begin{flalign}
	\frac{1}{2}\log \frac{ \det\left( \overline{\Sigma}_{\mathbf{v}} + \overline{\Sigma}_{\mathbf{n}} +\Lambda_{\mathbf{z}}
		\right)}{\det \left( \overline{\Sigma}_{\mathbf{v}} + \overline{\Sigma}_{\mathbf{n}} 
		\right)}
	=\frac{1}{2}\log \frac{ \det\left( \overline{\Sigma}_{\mathbf{v}} + \sigma_{\mathbf{v}}^2 I +\Lambda_{\mathbf{z}}
		\right)}{\det \left( \overline{\Sigma}_{\mathbf{v}} + \sigma_{\mathbf{v}}^2 I 
		\right)}
	. \nonumber
	\end{flalign}
	On the other hand, it may be verified (see Lemma~3.2 of \cite{fang2017towards}) that
	\begin{flalign}
	\frac{1}{2}\log \frac{ \det\left( \overline{\Sigma}_{\mathbf{v}} + \sigma_{\mathbf{v}}^2 I +\Lambda_{\mathbf{z}}
		\right)}{\det \left( \overline{\Sigma}_{\mathbf{v}} + \sigma_{\mathbf{v}}^2 I 
		\right)}
	\geq \frac{1}{2} \log  \prod_{i=1}^m \left[\frac{ \overline{\sigma}_{\mathbf{n} \left( i \right)}^2 + \sigma_{\mathbf{v}}^2 +\lambda_{i}}{\overline{\sigma}_{\mathbf{n} \left( i \right)}^2 + \sigma_{\mathbf{v}}^2} \right], \nonumber
	\end{flalign}
	where $\overline{\sigma}_{\mathbf{n} \left( i \right)}^2, i=1,\ldots,m$, are the diagonal terms of $\overline{\Sigma}_{\mathbf{n} }$, and the equality holds if $\overline{\Sigma}_{\mathbf{n} }$ is diagonal. Meanwhile, when $\overline{\Sigma}_{\mathbf{n} }$ is diagonal, let us denote
	\begin{flalign}
	\overline{\Sigma}_{\mathbf{n} }=\mathrm{diag}\left(\overline{\sigma}_{\mathbf{n} \left(1\right)}^2,\ldots, \overline{\sigma}_{\mathbf{n} \left(m\right)}^2 \right)
	=\mathrm{diag}\left(N_{1},\ldots,N_{m}\right) \nonumber
	\end{flalign}
	for simplicity.
	Then, the problem of
	\begin{flalign}
		\inf_{p_{\mathbf{n}}}
		I \left(\mathbf{x}; \widehat{\mathbf{y}} \right) \nonumber
	\end{flalign}
	reduces to that of choosing $N_1,\ldots, N_m$ to minimize 
	\begin{flalign}
	\frac{1}{2} \log  \prod_{i=1}^m \left( \frac{ N_{i} + \sigma_{\mathbf{v}}^2 +\lambda_{i} }{N_{i} + \sigma_{\mathbf{v}}^2} \right)
	= \sum_{i=1}^{m} \frac{1}{2}\log \left( 1+\frac{\lambda_{i}}{N_{i} + \sigma_{\mathbf{v}}^2} \right) \nonumber
	\end{flalign}
	subject to the constraint that 
	\begin{flalign}
	\sum_{i=1}^{m} N_{i} 
	= \mathrm{tr} \left( \overline{\Sigma}_{\mathbf{z}} \right)  
	= N. \nonumber
	\end{flalign}
	Define the Lagrange function by
	\begin{flalign}
	\sum_{i=1}^{m} \frac{1}{2}\log \left( 1+\frac{\lambda_{i}}{N_{i} + \sigma_{\mathbf{v}}^2} \right)+\zeta \left(\sum_{i=1}^{m} N_{i}-N\right), \nonumber
	\end{flalign}
	and differentiate it with respect to $N_{i}$, then we have
	\begin{flalign}
	\frac{\log \mathrm{e}}{2}\left( \frac{1}{N_{i}+ \sigma_{\mathbf{v}}^2+ \lambda_{i}}-\frac{1}{N_{i}+ \sigma_{\mathbf{v}}^2}\right) +\zeta
	=0, \nonumber
	\end{flalign}
	or equivalently,
	\begin{flalign}
	N_{i}
	=\frac{\sqrt{\lambda_{i}^2+\zeta \lambda_{i}}-\lambda_{i}}{2} - \sigma_{\mathbf{v}}^2
	=\frac{\eta}{2\left(1 + \sqrt{1+\frac{\eta }{\lambda_{i}}}
		\right)} - \sigma_{\mathbf{v}}^2, \nonumber
	\end{flalign}
	where $\eta= 2 \log \mathrm{e}  / \zeta \geq 0$.
	However, since $N_i \geq 0,~i=1, \ldots, m$, it may not always be possible
	to find a solution of this form; in other words, the term 
	\begin{flalign}
	\frac{\eta}{2\left(1 + \sqrt{1+\frac{\eta }{\lambda_{i}}}
		\right)} - \sigma_{\mathbf{v}}^2 \nonumber
	\end{flalign}
	may be negative for some $i$,  rendering this solution infeasible. Instead, we can use the Kuhn--Tucker
	conditions to verify that the optimal solution is in fact given by
	\begin{flalign}
	N_{i}
	= \left[ \frac{\eta}{2\left(1 + \sqrt{1+\frac{\eta }{\lambda_{i}}}
		\right)} - \sigma_{\mathbf{v}}^2 \right]^{+}, \nonumber
	\end{flalign}
	where
	\begin{flalign}
	\left[ x \right]^{+}
	= \left\{ \begin{array}{rcl}
	x & \text{if}~x > 0;\\
	0 & \text{if}~x\leq 0,
	\end{array} \right. \nonumber
	\end{flalign} 	
	and $\eta  \geq 0$ satisfies 
	\begin{flalign}
	\sum_{i=1}^{m} N_{i} 
	= \sum_{i=1}^{m} \left[ \frac{\eta}{2\left(1 + \sqrt{1+\frac{\eta }{\lambda_{i}}}
		\right)} - \sigma_{\mathbf{v}}^2 \right]^{+}
	= N. \nonumber
	\end{flalign}
	
	Consider now a scalar (dynamic) channel
	\begin{flalign}
	\widehat{\mathbf{y}}_{k} = \mathbf{z}_{k} + \mathbf{v}_{k} + \mathbf{n}_{k}, \nonumber
	\end{flalign}
	where $ \widehat{\mathbf{y}}_{k}, \mathbf{z}_{k}, \mathbf{v}_{k}, \mathbf{n}_{k} \in  \mathbb{R}$, while $ \left\{ \mathbf{z}_{k} \right\}$, $ \left\{ \mathbf{v}_{k} \right\}$, and $ \left\{ \mathbf{n}_{k} \right\}$ are mutually independent. In addition, $\left\{ \mathbf{z}_{k} \right\}$ is stationary colored Gaussian with power spectrum $S_{\mathbf{z}} \left( \omega \right)$, $\left\{ \mathbf{v}_{k} \right\}$ is stationary white Gaussian with variance $\sigma_{\mathbf{v}}^2$, and the noise power constraint is given by $\mathbb{E} \left[
	\mathbf{n}^{2}_{k} \right] \leq D$.
	%
	We may then consider a block of consecutive uses from time $0$ to $k$ of this channel 
	as $k+1$ channels in parallel \cite{Cov:06}. Particularly, let the eigen-decomposition of $\Sigma_{\mathbf{z}_{0,\ldots,k}}$ be given by
	\begin{flalign}
	\Sigma_{\mathbf{z}_{0,\ldots,k}}=U_{\mathbf{z}_{0,\ldots,k}}\Lambda_{\mathbf{z}_{0,\ldots,k}}U^{\mathrm{T}}_{\mathbf{z}_{0,\ldots,k}}, \nonumber
	\end{flalign} 
	where
	\begin{flalign}
	\Lambda_{\mathbf{z}_{0,\ldots,k}}
	=\mathrm{diag} \left( \lambda_{0},\ldots,\lambda_{k} \right). \nonumber
	\end{flalign}
	Then, we have
	\begin{flalign}
	&\min_{p_{\mathbf{n}_{0,\ldots,k}}:~\sum_{i=0}^{k} \mathbb{E}
		\left[ \mathbf{n}_{i}^{2} \right]
		\leq \left(k+1\right)D} \frac{I \left(\mathbf{x}_{0,\ldots,k}; \widehat{\mathbf{y}}_{0,\ldots,k} \right)}{k+1} \nonumber \\
	&~~~~ =\frac{1}{k+1} \sum_{i=0}^{k} \frac{1}{2}\log \left( 1+\frac{\lambda_{i}}{N_{i} + \sigma_{\mathbf{v}}^2 }\right), \nonumber
	\end{flalign}
	where
	\begin{flalign}
	N_{i}   = \left[ \frac{\eta}{2\left(\sqrt{1+\frac{\eta }{\lambda_{i}   }}+1
		\right)} - \sigma_{\mathbf{v}}^2 \right]^{+},~i=0,\ldots,k. \nonumber
	\end{flalign}
	Herein, $\eta \geq 0$ satisfies
	\begin{flalign}
	\sum_{i=0}^{k} N_{i}  
	= \sum_{i=0}^{k} \left[ \frac{\eta}{2\left(\sqrt{1+\frac{\eta }{\lambda_{i}   }}+1
		\right)} - \sigma_{\mathbf{v}}^2 \right]^{+}
	= \left( k+1 \right) N, \nonumber
	\end{flalign}
	or equivalently,
	\begin{flalign}
	\frac{1}{k+1} \sum_{i=0}^{k} N_{i}  
	= \frac{1}{k+1} \left[ \frac{\eta}{2\left(\sqrt{1+\frac{\eta }{\lambda_{i}   }}+1
		\right)} - \sigma_{\mathbf{v}}^2 \right]^{+}
	= N. \nonumber
	\end{flalign}
	In addition, as $k \to \infty$, the processes $ \left\{ \mathbf{z}_{k} \right\}$, $ \left\{ \mathbf{v}_{k} \right\}$,  $ \left\{ \mathbf{n}_{k} \right\}$, and $ \left\{ \widehat{\mathbf{y}}_{k} \right\}$ are stationary, and
	\begin{flalign}
	&\lim_{k\to \infty} \min_{p_{\mathbf{n}_{0,\ldots,k}}:~\sum_{i=0}^{k}
		\mathbb{E}
		\left[ \mathbf{n}_{i}^{2} \right] \leq \left(k+1\right)D} \frac{I \left(\mathbf{x}_{0,\ldots,k}; \widehat{\mathbf{y}}_{0,\ldots,k} \right)}{k+1} \nonumber \\
	&~~~~ =\inf_{p_{\mathbf{n}}}  \lim_{k\to \infty}\frac{I \left(\mathbf{x}_{0,\ldots,k}; \widehat{\mathbf{y}}_{0,\ldots,k} \right)}{k+1}\nonumber \\
	&~~~~=\inf_{p_{\mathbf{n}}}  \limsup_{k\to \infty}\frac{I \left(\mathbf{x}_{0,\ldots,k}; \widehat{\mathbf{y}}_{0,\ldots,k} \right)}{k+1}
	= \inf_{p_{\mathbf{n}}}I_{\infty} \left(\mathbf{x}; \widehat{\mathbf{y}} \right) 
	. \nonumber
	\end{flalign}
	On the other hand, since the processes are stationary, the covariance
	matrices are Toeplitz \cite{grenander1958toeplitz}, and their eigenvalues approach their limits as $k
	\to \infty$. Moreover, the densities of eigenvalues on the real line
	tend to the power spectra of the processes \cite{gutierrez2008asymptotically}. Accordingly,
	\begin{flalign}
	&\inf_{p_{\mathbf{n}}}I_{\infty} \left(\mathbf{x}; \widehat{\mathbf{y}} \right) \nonumber \\
	&~~~~ =\lim_{k\to \infty} \min_{p_{\mathbf{n}_{0,\ldots,k}}:~\sum_{i=0}^{k}
		\mathbb{E}
		\left[ \mathbf{n}_{i}^{2} \right] \leq \left(k+1\right) D} \frac{I \left(\mathbf{x}_{0,\ldots,k}; \widehat{\mathbf{y}}_{0,\ldots,k} \right)}{k+1}  \nonumber \\
	&~~~~ = \lim_{k\to \infty} \frac{1}{k+1} \sum_{i=0}^{k}\frac{1}{2}  \log \left( 1+\frac{\lambda_{i}}{N_{i} + \sigma_{\mathbf{v}}^2} \right)
	\nonumber \\
	&~~~~ = \frac{1}{2\pi} \int_{-\pi}^{\pi} \frac{1}{2}  \log \left[ 1+\frac{S_{\mathbf{z}} \left(\omega \right)}{N \left(\omega \right) + \sigma_{\mathbf{v}}^2} \right] \mathrm{d} \omega \nonumber \\
	&~~~~ = \frac{1}{2\pi} \int_{-\pi}^{\pi} \log \sqrt{ 1+\frac{S_{\mathbf{z}} \left(\omega \right)}{N \left(\omega \right) + \sigma_{\mathbf{v}}^2} } \mathrm{d} \omega,
	\nonumber
	\end{flalign}
	where
	\begin{flalign}
	N \left(\omega \right) 
	= \left\{ \frac{\eta}{2\left[\sqrt{1+\frac{\eta }{S_{\mathbf{z}} \left( \omega \right)}}+1 \right]} - \sigma_{\mathbf{v}}^2 \right\}^{+}, \nonumber
	\end{flalign}
	and $\eta \geq 0 $ satisfies
	\begin{flalign}
	\lim_{k\to \infty} \frac{1}{k+1}\sum_{i=0}^{k} N_{i}
	=\frac{1}{2\pi} \int_{-\pi}^{\pi} N \left(\omega \right) \mathrm{d}  \omega \nonumber 
	=N. \nonumber
	\end{flalign}
	This concludes the proof.
\end{proof}

In general, it can be verified that the more distortion allowed, the less privacy leakage will occur. This privacy-distortion tradeoff is analytically captured in Theorem~\ref{privacy1}. In the extreme case of when $\left\{ \mathbf{z}_{k} \right\}$ is  stationary white Gaussian, that is, when $A = 0$, we have
\begin{flalign} 
\sigma_{\mathbf{z}}^2
= S_{\mathbf{z}} \left( \omega \right)
=	C  W  C^{\mathrm{T}},
\end{flalign} 
and the privacy-distortion tradeoff in Theorem~\ref{privacy1} reduces to 
\begin{flalign} 
\inf_{\mathbb{E} \left[ \left( \mathbf{y}_k - \widehat{\mathbf{y}}_{k} \right)^2 \right] \leq D} I_{\infty} \left( \mathbf{x} ; \widehat{\mathbf{y}} \right) 
= \frac{1}{2} \log \left( 1 + \frac{C  W  C^{\mathrm{T}}}{D + \sigma_{\mathbf{v}}^2} \right).
\end{flalign} 

Note also that in general $N \left( \omega \right)$ becomes larger as $S_{\mathbf{z}} \left( \omega \right)$ becomes larger in \eqref{spectrum1}; particularly, it may be verified that when $S_{\mathbf{z}} \left( \omega \right)$ is below the threshold 
\begin{flalign}
	\frac{\eta}{\left( \frac{\eta}{2 \sigma_{\mathbf{v}}^2} - 1 \right)^2 - 1},
\end{flalign}
then $N \left( \omega \right) = 0$; while when $S_{\mathbf{z}} \left( \omega \right)$ is above the aforementioned threshold, $N \left( \omega \right)$ strictly increases with $S_{\mathbf{z}} \left( \omega \right)$. This means that more power shall be delivered to frequencies at which the power spectra of $\left\{ \mathbf{z}_{k} \right\}$ are larger (above a threshold). In a broad sense, this solution may be viewed as a ``thresholded obfuscating" power allocation policy (cf. discussions  in \cite{CLarxiv} on ``obfuscating" power allocation solutions, as well as the relations with ``water-filling" and ``reverse water-filling" policies).



\subsection{Perspective of ``Virtual Channel"} \label{virtual}

In fact, the system in Fig.~\ref{system2} may be viewed as a ``virtual channel" modeled as
\begin{flalign}
\widehat{\mathbf{y}}_{k} 
= \mathbf{y}_{k} + \mathbf{n}_{k}
= \mathbf{z}_{k} +
\mathbf{v}_{k} + \mathbf{n}_{k},
\end{flalign}
where $\left\{ \mathbf{z}_{k} \right\}$ (or equivalently, $\left\{ \mathbf{x}_{k} \right\}$; see \eqref{mutual}) is the channel input, $\left\{ \widehat{\mathbf{y}}_{k} \right\}$ is the channel output, $\left\{ \mathbf{v}_{k} \right\}$ is the channel noise that is pre-given and cannot be designed, and $\left\{ \mathbf{n}_{k} \right\}$ is the channel noise that can be designed (subject to a constraint). This channel model may be viewed as a generalized version of that considered in \cite{CLarxiv}; particularly, the leakage of this channel is measured by
\begin{flalign} \label{mutual}
I_{\infty} \left( \mathbf{z}; \widehat{\mathbf{y}} \right) 
= I_{\infty} \left( \mathbf{x}; \widehat{\mathbf{y}} \right),
\end{flalign}
subject to a power constraint
\begin{flalign}
\mathbb{E} \left( \mathbf{n}_k^2 \right) = \mathbb{E} \left[ \left( \mathbf{y}_k - \widehat{\mathbf{y}}_{k} \right)^2 \right] \leq D.
\end{flalign} 
Note that herein we have employed the following steps to prove \eqref{mutual}:
\begin{flalign}
&I_{\infty} \left( \mathbf{z}; \widehat{\mathbf{y}} \right) 
= h_{\infty} \left( \widehat{\mathbf{y}} \right) - h_{\infty} \left( \widehat{\mathbf{y}} | \mathbf{z} \right) \nonumber \\
&= h_{\infty} \left( \widehat{\mathbf{y}} \right) - h_{\infty} \left( \mathbf{z} + \mathbf{v} + \mathbf{n} | \mathbf{z} \right) 
= h_{\infty} \left( \widehat{\mathbf{y}} \right) - h_{\infty} \left( \mathbf{v} + \mathbf{n} | \mathbf{z} \right) \nonumber \\
&= h_{\infty} \left( \widehat{\mathbf{y}} \right) - h_{\infty} \left( \mathbf{v} + \mathbf{n} \right) 
= h_{\infty} \left( \widehat{\mathbf{y}} \right) - h_{\infty} \left( \mathbf{v} + \mathbf{n} | \mathbf{x} \right) \nonumber \\
&= h_{\infty} \left( \widehat{\mathbf{y}} \right) - h_{\infty} \left( \mathbf{z} + \mathbf{v} + \mathbf{n} | \mathbf{x} \right) 
= h_{\infty} \left( \widehat{\mathbf{y}} \right) - h_{\infty} \left( \widehat{\mathbf{y}} | \mathbf{x} \right) \nonumber \\
&= I_{\infty} \left( \mathbf{x}; \widehat{\mathbf{y}} \right).
\end{flalign}



\subsection{Connection with Conditional Entropy}

Note first that 
\begin{flalign}
I_{\infty} \left( \mathbf{x} ; \widehat{\mathbf{y}} \right)
= h_{\infty} \left( \mathbf{x} \right) - h_{\infty} \left( \mathbf{x} | \widehat{\mathbf{y}} \right),
\end{flalign}
and hence
\begin{flalign}
&h_{\infty} \left( \mathbf{x} | \widehat{\mathbf{y}} \right)
= h_{\infty} \left( \mathbf{x} \right) - I_{\infty} \left( \mathbf{x} ; \widehat{\mathbf{y}} \right) \nonumber \\ 
&~~~~ = \frac{1}{2 \pi} \int_{0}^{2 \pi} \log \sqrt{ \left( 2 \pi \mathrm{e} \right)^m \det \Phi_{\mathbf{x}} \left( \omega \right)} \mathrm{d} \omega - I_{\infty} \left( \mathbf{x} ; \widehat{\mathbf{y}} \right).
\end{flalign}
Since $\Phi_{\mathbf{x}} \left( \omega \right)$ is pre-given as
\begin{flalign}
\Phi_{\mathbf{x}} \left( \omega \right)
= \left( \mathrm{e}^{\mathrm{j} \omega} I - A \right)^{-1} W \left( \mathrm{e}^{- \mathrm{j} \omega} I - A \right)^{-\mathrm{T}},
\end{flalign} 
minimizing $I_{\infty} \left( \mathbf{x} ; \widehat{\mathbf{y}} \right)$ is in fact equivalent to maximizing $h_{\infty} \left( \mathbf{x} | \widehat{\mathbf{y}} \right)$, which is another privacy measure that is oftentimes employed in estimation problems (see, e.g., \cite{Cov:06, FangITW19} and the references therein). Particularly, it holds that
\begin{flalign}
&\sup_{\mathbb{E} \left[ \left( \mathbf{y}_k - \widehat{\mathbf{y}}_{k} \right)^2 \right] \leq D} h_{\infty} \left( \mathbf{x} | \widehat{\mathbf{y}} \right) \nonumber \\ 
&~~~~ = \frac{1}{2 \pi} \int_{0}^{2 \pi} \log \sqrt{ \left( 2 \pi \mathrm{e} \right)^m \det \Phi_{\mathbf{x}} \left( \omega \right)} \mathrm{d} \omega \nonumber \\
&~~~~~~~~ - \frac{1}{2 \pi} \int_{0}^{2 \pi} \log \sqrt{ 1 + \frac{S_{\mathbf{z}} \left( \omega \right)}{N \left( \omega \right) + \sigma_{\mathbf{v}}^2} } \mathrm{d} \omega. \nonumber \\
&~~~~= \frac{1}{2 \pi} \int_{0}^{2 \pi} \log \sqrt{ \left( 2 \pi \mathrm{e} \right)^m \frac{\left[ N \left( \omega \right) + \sigma_{\mathbf{v}}^2 \right]\det \Phi_{\mathbf{x}} \left( \omega \right) }{S_{\mathbf{z}} \left( \omega \right) + N \left( \omega \right) + \sigma_{\mathbf{v}}^2}} \mathrm{d} \omega,
\end{flalign} 
where $N \left( \omega \right)$ is given by \eqref{spectrum1}.

\subsection{Comparison with Adding I.I.D. Gaussian Masks}

What is the difference between the solution in \eqref{leakage1} and adding stationary white (i.i.d.) Gaussian privacy masks instead? It can be verified that in the i.i.d. case, the information leakage rate subject to distortion constraint 
\begin{flalign}
\mathbb{E} \left[ \left( \mathbf{y}_k - \widehat{\mathbf{y}}_{k} \right)^2 \right] \leq D
\end{flalign}
is given by
\begin{flalign}
I_{\infty} \left( \mathbf{x} ; \widehat{\mathbf{y}} \right)  = \frac{1}{2 \pi} \int_{0}^{2 \pi} \log \sqrt{ 1 + \frac{S_{\mathbf{z}} \left( \omega \right)}{D + \sigma_{\mathbf{v}}^2} } \mathrm{d} \omega.
\end{flalign} 
In comparison with \eqref{leakage1}, it may be verified that
\begin{flalign}
&\frac{1}{2 \pi} \int_{0}^{2 \pi} \log \sqrt{ 1 + \frac{S_{\mathbf{z}} \left( \omega \right)}{D + \sigma_{\mathbf{v}}^2} } \mathrm{d} \omega \nonumber \\
&~~~~ \geq \frac{1}{2 \pi} \int_{0}^{2 \pi} \log \sqrt{ 1 + \frac{S_{\mathbf{z}} \left( \omega \right)}{N \left( \omega \right) + \sigma_{\mathbf{v}}^2} } \mathrm{d} \omega,
\end{flalign} 
where equality holds if and only if $S_{\mathbf{z}} \left( \omega \right) = 0$. That is to say, when subject to the same distortion constraint, adding i.i.d. Gaussian privacy masks will always lead to more privacy leakage than adding stationary colored Gaussian privacy masks with power spectra shaped according to \eqref{spectrum1}.

\subsection{Dual Problem}

In fact, the question Theorem~\ref{privacy1} answers is: Given a certain distortion constraint, what is the minimum privacy leakage (and how to design the optimal privacy mask)?
On the other hand, the dual problem would be: Given a certain privacy level, what is the minimum degree of distortion (and how to design the optimal privacy mask)? The following corollary answers the latter question.

\begin{corollary} \label{corollary1}
	Consider the dynamical system with privacy masks depicted in Fig.~\ref{system2}. Suppose that the properties of $\left\{ \mathbf{n}_{k} \right\}$ can be designed. 
	Then, in order to make sure that the information leakage is upper bounded by a constant $R > 0$ as
	\begin{flalign}
	I_{\infty} \left( \mathbf{x} ; \widehat{\mathbf{y}} \right) \leq R,
	\end{flalign}
	the minimum distortion between $\left\{ \mathbf{y}_{k} \right\}$ and $\left\{ \widehat{\mathbf{y}}_{k} \right\}$ is given by
	\begin{flalign}
	&\inf_{ I_{\infty} \left( \mathbf{x} ; \widehat{\mathbf{y}} \right) \leq R } \mathbb{E} \left[ \left( \mathbf{y}_k - \widehat{\mathbf{y}}_{k} \right)^2 \right]
	= \frac{1}{2 \pi} \int_{0}^{2 \pi} 
	N \left( \omega \right) \mathrm{d} \omega \nonumber \\
	&~~~~ = \frac{1}{2 \pi} \int_{0}^{2 \pi} \left\{ \frac{\eta}{2 \left[ 1 + \sqrt{1 + \frac{\eta}{S_{\mathbf{z}} \left( \omega \right)} } \right]} - \sigma_{\mathbf{v}}^2 \right\}^{+} \mathrm{d} \omega,
	\end{flalign} 
	where $\eta \geq 0$ satisfies
	\begin{flalign}
	\frac{1}{2 \pi} \int_{0}^{2 \pi} \log \sqrt{ 1 + \frac{S_{\mathbf{z}} \left( \omega \right)}{N \left( \omega \right) + \sigma_{\mathbf{v}}^2}} \mathrm{d} \omega 
	= R.
	\end{flalign} 
	Herein, 
	\begin{flalign} \label{spectrum2}
	N \left( \omega \right) = \left\{ \frac{\eta}{2 \left[ 1 + \sqrt{1 + \frac{\eta}{S_{\mathbf{z}} \left( \omega \right)} } \right]} - \sigma_{\mathbf{v}}^2 \right\}^{+}.
	\end{flalign} 
	Furthermore, in order to achieve this minimum distortion, the noise
	$\left\{ \mathbf{n}_{k} \right\}$ should be chosen as a stationary colored Gaussian process with power spectrum
	\eqref{spectrum2}.
\end{corollary}  

\begin{proof}
	The proof follows steps similar to those in the proof of Theorem~\ref{privacy1} in a dual manner. 
\end{proof}


Note that in the extreme case of when $\left\{ \mathbf{z}_{k} \right\}$ is stationary white Gaussian, that is, when $A = 0$, Corollary~\ref{corollary1} reduces to 
\begin{flalign} 
\inf_{ I_{\infty} \left( \mathbf{x} ; \widehat{\mathbf{y}} \right) \leq R } \mathbb{E} \left[ \left( \mathbf{y}_k - \widehat{\mathbf{y}}_{k} \right)^2 \right]
= \frac{C  W  C^{\mathrm{T}}}{2^{2 R } - 1} - \sigma_{\mathbf{v}}^2.
\end{flalign} 

Equivalently, the tradeoffs captured in Theorem~\ref{privacy1} and Corollary~\ref{corollary1} can instead be expressed using the Lagrangian formulation as 
\begin{flalign}
\inf_{p_{\mathbf{n}}} 
\left\{ I_{\infty} \left( \mathbf{x} ; \widehat{\mathbf{y}} \right) 
+ \alpha \mathbb{E} \left[ \left( \mathbf{y}_k - \widehat{\mathbf{y}}_{k} \right)^2 \right] 
\right\}
,
\end{flalign} 
or
\begin{flalign}
\inf_{p_{\mathbf{n}}} 
\left\{ \mathbb{E} \left[ \left( \mathbf{y}_k - \widehat{\mathbf{y}}_{k} \right)^2 \right] 
+ \beta I_{\infty} \left( \mathbf{x} ; \widehat{\mathbf{y}} \right) 
\right\}
,
\end{flalign} 
where $\alpha, \beta > 0$ denote the tradeoff parameters.

\subsection{Output Power Constraint}

Consider next the case of output power constraint.

\begin{corollary} \label{privacy2}
	Consider the dynamical system with privacy masks depicted in Fig.~\ref{system2}. Suppose that the properties of $\left\{ \mathbf{n}_{k} \right\}$ can be designed subject to a masked output power constraint 
	\begin{flalign}
	\mathbb{E} \left[ \widehat{\mathbf{y}}_k^2 \right] \leq \widehat{Y}.
	\end{flalign}
	Then, in order to minimize the information leakage rate
	\begin{flalign}
	I_{\infty} \left( \mathbf{x} ; \widehat{\mathbf{y}} \right),
	\end{flalign}
	the noise
	$\left\{ \mathbf{n}_{k} \right\}$ should be chosen as a stationary colored Gaussian process. In addition, the power spectrum of $\left\{ \mathbf{n}_{k} \right\}$ should be chosen as
	\begin{flalign}
	N \left( \omega \right) = \left\{ \frac{\eta}{2 \left[ 1 + \sqrt{1 + \frac{\eta}{S_{\mathbf{z}} \left( \omega \right)} } \right]} - \sigma_{\mathbf{v}}^2 \right\}^{+},
	\end{flalign} 
	where $\eta \geq 0$ satisfies
	\begin{flalign}
	&\frac{1}{2 \pi} \int_{0}^{2 \pi} N \left( \omega \right) \mathrm{d} \omega \nonumber \\
	&~~~~ = \frac{1}{2 \pi} \int_{0}^{2 \pi} \left\{ \frac{\eta}{2 \left[ 1 + \sqrt{1 + \frac{\eta}{S_{\mathbf{z}} \left( \omega \right)} } \right]} - \sigma_{\mathbf{v}}^2 \right\}^{+} \mathrm{d} \omega \nonumber \\
	&~~~~ = \widehat{Y} - \frac{1}{2 \pi} \int_{0}^{2 \pi} S_{\mathbf{z}} \left( \omega \right) \mathrm{d} \omega - \sigma_{\mathbf{v}}^2.
	\end{flalign} 
	Correspondingly, the minimum information leakage rate is given by
	\begin{flalign}
	\inf_{\mathbb{E} \left[ \widehat{\mathbf{y}}_k^2 \right] \leq \widehat{Y}} I_{\infty} \left( \mathbf{x} ; \widehat{\mathbf{y}} \right)
	= \frac{1}{2 \pi} \int_{0}^{2 \pi} \log \sqrt{ 1 + \frac{S_{\mathbf{z}} \left( \omega \right)}{N \left( \omega \right) + \sigma_{\mathbf{v}}^2} } \mathrm{d} \omega.
	\end{flalign} 
\end{corollary}

\vspace{3mm}

\begin{proof}
	Since $\left\{ \mathbf{z}_{k} \right\}$, $\left\{ \mathbf{v}_{k} \right\}$, and $\left\{ \mathbf{n}_{k} \right\}$ are mutually independent, we have 
	\begin{flalign}
	S_{\widehat{\mathbf{y}}} \left( \omega \right) 
	&= S_{\mathbf{z} + \mathbf{v} + \mathbf{n}} \left( \omega \right) 
	= S_{\mathbf{z}} \left( \omega \right) + S_{\mathbf{v}} \left( \omega \right) + N \left( \omega \right) \nonumber \\
	& = S_{\mathbf{z}} \left( \omega \right) + \sigma_{\mathbf{v}}^2 + N \left( \omega \right), \nonumber
	\end{flalign} 
	and thus
	\begin{flalign}
	\mathbb{E} \left[ \mathbf{n}_k^2 \right]
	&= \frac{1}{2 \pi} \int_{0}^{2 \pi} N \left( \omega \right) \mathrm{d} \omega  \nonumber \\  
	&= \frac{1}{2 \pi} \int_{0}^{2 \pi} S_{\mathbf{y}} \left( \omega \right) \mathrm{d} \omega  - \frac{1}{2 \pi} \int_{0}^{2 \pi} S_{\mathbf{z}} \left( \omega \right) \mathrm{d} \omega  - \sigma_{\mathbf{v}}^2 \nonumber \\
	&= \mathbb{E} \left[ \mathbf{y}_k^2 \right] - \frac{1}{2 \pi} \int_{0}^{2 \pi} S_{\mathbf{z}} \left( \omega \right) \mathrm{d} \omega - \sigma_{\mathbf{v}}^2 \nonumber \\
	&\leq Y - \frac{1}{2 \pi} \int_{0}^{2 \pi} S_{\mathbf{z}} \left( \omega \right) \mathrm{d} \omega - \sigma_{\mathbf{v}}^2. \nonumber
	\end{flalign} 
	Then, Corollary~\ref{privacy2} follows by invoking Theorem~\ref{privacy1}.
\end{proof}

We may again consider the following dual problem.

\begin{corollary}
	Consider the dynamical system with privacy masks depicted in Fig.~\ref{system2}. Suppose that the properties of $\left\{ \mathbf{n}_{k} \right\}$ can be designed.
	Then, in order to make sure that the information leakage is upper bounded by a constant $R > 0$ as
	\begin{flalign}
	I_{\infty} \left( \mathbf{x} ; \widehat{\mathbf{y}} \right) \leq R,
	\end{flalign}
	the minimum power of the masked data $\left\{ \widehat{\mathbf{y}}_{k} \right\}$ is given by
	\begin{flalign}
	&\inf_{ I_{\infty} \left( \mathbf{x} ; \widehat{\mathbf{y}} \right) \leq R }\mathbb{E} \left[ \widehat{\mathbf{y}}_{k}^2 \right] \nonumber \\
	&~~~~ = \frac{1}{2 \pi} \int_{0}^{2 \pi} N \left( \omega \right) \mathrm{d} \omega + \frac{1}{2 \pi} \int_{0}^{2 \pi} S_{\mathbf{z}} \left( \omega \right) \mathrm{d} \omega + \sigma_{\mathbf{v}}^2 \nonumber \\
	&~~~~ = \frac{1}{2 \pi} \int_{0}^{2 \pi} \left\{ \frac{\eta}{2 \left[ 1 + \sqrt{1 + \frac{\eta}{S_{\mathbf{z}} \left( \omega \right)} } \right]} - \sigma_{\mathbf{v}}^2 \right\}^{+} \mathrm{d} \omega  \nonumber \\
	&~~~~~~~~ + \frac{1}{2 \pi} \int_{0}^{2 \pi} S_{\mathbf{z}} \left( \omega \right) \mathrm{d} \omega + \sigma_{\mathbf{v}}^2,
	\end{flalign} 
	where 
	\begin{flalign}
	N \left( \omega \right) = \left\{ \frac{\eta}{2 \left[ 1 + \sqrt{1 + \frac{\eta}{S_{\mathbf{z}} \left( \omega \right)} } \right]} - \sigma_{\mathbf{v}}^2 \right\}^{+},
	\end{flalign}
	and
	$\eta \geq 0$ satisfies
	\begin{flalign}
	\frac{1}{2 \pi} \int_{0}^{2 \pi} \log \sqrt{ 1 + \frac{S_{\mathbf{z}} \left( \omega \right)}{N \left( \omega \right) + \sigma_{\mathbf{v}}^2} } \mathrm{d} \omega
	= R.
	\end{flalign} 
	Furthermore, in order to achieve this minimum distortion, the noise
	$\left\{ \mathbf{n}_{k} \right\}$ should be chosen as a stationary colored Gaussian process with power spectrum
	\begin{flalign}
	N \left( \omega \right) = \frac{\eta}{2 \left[ 1 + \sqrt{1 + \frac{\eta}{S_{\mathbf{x}} \left( \omega \right)}} \right]}.
	\end{flalign} 
\end{corollary}  

\vspace{3mm}

\section{CONCLUSIONS}

In this paper, we have derived analytical formulas for the fundamental limits of obfuscation in terms of privacy-distortion tradeoffs for linear Gaussian dynamical systems with an  information-theoretic analysis. In addition, we have also obtained explicit ``thresholded obfuscating" power allocation solutions on how to design the optimal privacy masks.

Potential future research directions include the analysis of non-Gaussian noises, as well as  investigating the implications of the results in the context of state estimation and feedback control systems.

\balance

\bibliographystyle{IEEEtran}
\bibliography{references}




\end{document}